\documentclass[11pt]{article}

\usepackage{amsmath, amssymb}
\usepackage{amsthm}
\usepackage{algorithm}
\usepackage{algpseudocode} 
\usepackage{graphicx}
\usepackage{hyperref}
\usepackage{cite}
\usepackage{geometry}
\usepackage{thm-restate}
\newtheorem{claim}{Claim}
\newtheorem{lemma}[claim]{Lemma}
\newtheorem{theorem}[claim]{Theorem}
\newtheorem*{theorem*}{Theorem}

\title{\textbf{A Bottom-Up Algorithm for Negative-Weight \\ SSSP with Integrated Negative Cycle Finding}}
\author{
    Jason Li\thanks{Carnegie Mellon University, \texttt{jmli@cs.cmu.edu}} \and
    Connor Mowry\thanks{Carnegie Mellon University, \texttt{cmowry@andrew.cmu.edu}}
}
\date{}

\begin{document}

\maketitle

\begin{abstract}
We present a simplified algorithm for solving the Negative-Weight Single-Source Shortest Paths (SSSP) problem, focusing on enhancing clarity and practicality over prior methods. Our algorithm uses graph diameter as a recursive parameter, offering greater robustness to the properties of the decomposed graph compared to earlier approaches. Additionally, we fully integrate negative-weight cycle finding into the algorithm by augmenting the Bellman-Ford/Dijkstra hybrid, eliminating the need for a separate cycle-finding procedure found in prior methods. Although the algorithm achieves no theoretical efficiency gains, it simplifies negative cycle finding and emphasizes design simplicity, making it more accessible for implementation and analysis. This work highlights the importance of robust parameterization and algorithmic simplicity in addressing the challenges of Negative-Weight SSSP.
\end{abstract}

\section{Introduction}

The Single-Source Shortest Paths (SSSP) problem with negative edge weights is a cornerstone of algorithmic graph theory, with applications spanning network optimization, logistics, and resource allocation. An important extension of this problem involves handling graphs that may contain negative-weight cycles, where the goal is to either return a shortest path tree or find a negative-weight cycle. Both outputs serve as verifiable certificates, ensuring correctness and reliability.

Scaling has been a foundational technique in shortest path algorithms since the 1990s~\cite{goldberg95}. Initially, it was used to normalize edge weights, reducing the problem to one where all negative edges have weight exactly \(-1\). More recently, scaling has been applied to transform cycle properties, ensuring that non-negative-weight cycles become significantly more positive. This transformation enables efficient computation of SSSP on the scaled graph. Algorithms such as those in the Bernstein-Nanongkai-Wulff-Nilsen~\cite{bernstein2023negativeweightsinglesourceshortestpaths} framework have integrated scaling with advanced decomposition techniques, achieving near-linear time complexity for solving the Negative-Weight SSSP problem.

Building on this foundation, Bringmann, Cassis, and Fischer~\cite{bringmann2023negativeweightsinglesourceshortestpaths} extended the existing framework, achieving a nearly 6-log-factor improvement over~\cite{bernstein2023negativeweightsinglesourceshortestpaths}. Their work set a new standard for solving Negative-Weight SSSP, providing both theoretical advancements and practical insights.

In this paper, we build upon the techniques introduced by~\cite{bringmann2023negativeweightsinglesourceshortestpaths}, presenting a \textit{bottom-up algorithm} inspired by~\cite{fischer2025simple} with the following contributions:
\begin{enumerate}
    \item \textbf{Diameter-Based Decomposition Parameter}: Our algorithm employs the \emph{diameter} of the non-negative graph as a parameter in our recursive decomposition, similar to the bottom-up approach of~\cite{fischer2025simple}. Compared to~\cite{bernstein2023negativeweightsinglesourceshortestpaths} and~\cite{bringmann2023negativeweightsinglesourceshortestpaths}'s parameterization of the maximum number of negative-weight edges on any shortest path, this parameterization is more natural and, more importantly, avoids the issue of undetected failures in the presence of negative-weight cycles. In particular, we bypass the noisy binary search analysis of~\cite{bringmann2023negativeweightsinglesourceshortestpaths}, which they must perform to handle the case of undetected failures.
    \item \textbf{Integrated Negative Cycle Finding}: The diameter-based decomposition allows us to augment the Bellman-Ford/Dijkstra hybrid algorithm to look for negative-weight cycles on the fly. In other words, we fully integrate negative-weight cycle finding into our algorithm, instead of designing a separate procedure as in~\cite{bringmann2023negativeweightsinglesourceshortestpaths}.
\end{enumerate}

Although our algorithm does not improve upon the theoretical efficiency of~\cite{bringmann2023negativeweightsinglesourceshortestpaths}, it emphasizes simplicity in design and implementation, making it an accessible and practical alternative for solving the Negative-Weight SSSP problem.

\begin{theorem}\label{thm:main}
Consider a graph with integral edge weights that are at least $-W$. There is a Las Vegas algorithm that solves Negative-Weight SSSP in $O((m+n\log\log n)\log^2n\log(nW))$ time with high probability.
\end{theorem}

\subsection{Preliminaries}

Let $G=(V,E,w_G)$ be a directed graph with integral weight function $w_G:E\to\mathbb Z$.
A \emph{potential function} $\phi:V\to\mathbb R$ is a function on the vertices used to reweight the graph: the reweighted graph $G_\phi$ satisfies $w_{G_\phi}(uv)=w_G(uv)+\phi(u)-\phi(v)$ for all edges $uv\in E$. A result of Johnson~\cite{johnson1977efficient} reduces the Negative-Weight SSSP problem to computing a potential function $\phi$ such that all edge weights of $G_\phi$ are non-negative, after which Dijkstra's algorithm can be run on $G_\phi$ to recover the single-source shortest paths in $G$.

The iterative \emph{scaling technique} of~\cite{bernstein2023negativeweightsinglesourceshortestpaths} does not compute such a potential function in one go. Instead, if the graph $G$ has all edge weights at least $-W$, then the goal is to either compute a potential function $\phi$ such that $G_\phi$ has all edge weights at least $-W/2$, or output a negative-weight cycle; we call this problem \textsc{Scale}.~\cite{bernstein2023negativeweightsinglesourceshortestpaths} show that Negative-Weight SSSP reduces to $O(\log(nW))$ iterations of \textsc{Scale} when the original graph has all edge weights at least $-W$. Moreover,~\cite{bringmann2023negativeweightsinglesourceshortestpaths} show that if each iteration of \textsc{Scale} runs in $T$ time \emph{in expectation}, then Negative-Weight SSSP can be computed in $O(T\log(nW))$ time \emph{with high probability}\footnote{\emph{With high probability} means with probability at least $1-1/n^C$ for arbitrarily large constant $C>0$.} (and with Las Vegas guarantee) by restarting any runs of \textsc{Scale} that exceed $2T$ time and applying a Chernoff bound. Hence, to obtain Theorem~\ref{thm:main}, we only focus on the task of solving \textsc{Scale} in $O((m+n\log\log n)\log^2n)$ expected time. It even suffices to solve \textsc{Scale} in this time with high probability, since the solution can be easily checked for correctness and \textsc{Scale} can be restarted if incorrect.

We need three main subroutines from~\cite{bernstein2023negativeweightsinglesourceshortestpaths,bringmann2023negativeweightsinglesourceshortestpaths}. The first is the \textsc{Decompose} subroutine introduced in~\cite{bernstein2023negativeweightsinglesourceshortestpaths} and refined in~\cite{bringmann2023negativeweightsinglesourceshortestpaths}. Define the in-ball/out-ball of radius $r$ centered at $v$\ as the set of vertices reachable to/from  $v$ by a path of weight at most $r$, respectively. Given a graph $G$, define the graph $G_{\ge 0}$ as the graph $G$ with all negative-weight edges reweighted to $0$. We require the following guarantee from~\cite{bringmann2023negativeweightsinglesourceshortestpaths}:
\begin{lemma}[Lemma~8 of~\cite{bringmann2023negativeweightsinglesourceshortestpaths}]\label{lem:decompose}
There is a randomized algorithm \textsc{Decompose}$(G,d)$ running in expected time $O((m+n\log\log n)\log n)$ that computes a set $S\subseteq E$ of positive-weight edges such that: 
\begin{enumerate}
\item \textup{Progress:} With high probability, for any strongly connected component $C$ in $G\setminus S$, either (i) $|C|\le\frac34|V|$ or (ii) for each vertex $v\in C$, both the in-ball and out-ball of radius \( \frac{d}{4} \) centered at \( v \) contain more than \( \frac{|V|}{2} \) vertices.
\item \textup{Sparse Hitting:} For any (potentially non-simple) path $P$ of length at most $d$ in $G_{\ge 0}$, the number of edges of $P$ inside $S$ (counting multiplicity) has expectation at most $O(\log n)$.
\end{enumerate}
\end{lemma}
We remark that Lemma~8 of~\cite{bringmann2023negativeweightsinglesourceshortestpaths} proves a different condition~\emph{(ii)} for progress that depends on a parameter $\kappa$ that they introduce. However, our new condition~\emph{(ii)} also holds and is mentioned in the proof of their Lemma~22.

The second subroutine is a Bellman-Ford/Dijkstra hybrid algorithm for Negative-Weight SSSP, also introduced in~\cite{bernstein2023negativeweightsinglesourceshortestpaths} and refined in~\cite{bringmann2023negativeweightsinglesourceshortestpaths}. For convenience, we do not mention a source vertex in the input to \textsc{BellmanFordDijkstra}. Instead, the goal is to compute, for each vertex $v\in V$, the shortest path ending at $v$ (and starting from anywhere). Note that the shortest path has weight at most $0$ since the empty path at $v$ of weight $0$ is a possible choice. To frame this problem as SSSP, simply add a source vertex $s$ with zero-weight edges to each vertex in $V$, and call SSSP on the new graph $G_s$ with source $s$. For further convenience, we also integrate a potential function $\phi:V\to\mathbb R$ directly into \textsc{BellmanFordDijkstra}, which runs the Bellman-Ford/Dijkstra hybrid on $(G_s)_\phi$ instead (with the extension $\phi(s)=0$), and quickly finds single-source shortest paths with few negative-weight edges in $(G_s)_\phi$. These single-source shortest paths in $(G_s)_\phi$ correspond exactly to single-source shortest paths in $G_s$, which in turn correspond to shortest paths in $G$ ending at each vertex.

In the version below, we also include the algorithm's guarantees after each iteration $i$ to allow for early termination if a negative-weight cycle is detected. In addition, we also maintain the weights of our current shortest paths under a separate, auxiliary weight function $w'$. This can be implemented by recording the distance of each current path $P$ as a tuple $(w_G(P),w'(P))$, but only using the first coordinate for comparisons.

\begin{lemma}[Lemmas 25 and 62 of~\cite{bringmann2023negativeweightsinglesourceshortestpaths}]\label{lem:BFD}
Let $G=(V,E,w_G)$ be a directed graph and $\phi$ be a potential function. For each vertex $v\in V$, let $\eta_{G,\phi}(v)$ be the minimum number of edges of negative weight in $G_\phi$ among all shortest paths in $G$ ending at $v$. There is an algorithm \textsc{BellmanFordDijkstra}$(G,\phi)$ that computes a shortest path ending at each vertex in time $O(\sum_v(\deg(v)+\log\log n)\cdot\eta_{G,\phi}(v))$.

The first $i$ iterations of the algorithm runs in $O((m+n\log\log n)\cdot i)$ time and computes, for each vertex $v\in V$, the minimum weight $d_i(v)$ of a path $P_i(v)$ from $s$ to $v$ among those that contains less than $i$ negative-weight edges. Moreover, given input auxiliary weights $w'$ on the edges (independent of the weights $w_G$ that determine shortest paths), the algorithm can also compute the weight $d_i'(v)$ under $w'$ of some such path $P_i(v)$.
\end{lemma}

The third subroutine is given a graph with its strongly connected components and a potential function, and fixes the DAG edges outside the SCCs while keeping the edge weights inside the SCCs non-negative.
\begin{lemma}[Lemma 3.2 of~\cite{bernstein2023negativeweightsinglesourceshortestpaths}]\label{lem:fix-dag}
Consider a graph $G$ with strongly connected components $C_1,\ldots,C_\ell$ and a potential function $\phi$, and suppose that all edges inside the SCCs have non-negative weight in $G_\phi$. There is an $O(n+m)$ time algorithm that adjusts $\phi$ so that all edge weights in $G_\phi$ are now non-negative.
\end{lemma}

\section{The Scale Algorithm}

Recall the specifications of the \textsc{Scale} problem: given an input graph \( G = (V, E, w_G) \) with edge weights \( w_G(e) \geq -W \) for all \( e \in E \), either return a reweighted graph \( G_\phi \) with weights at least \( -W/2 \) or output a negative-weight cycle in \( G \). We present an algorithm that runs in expected time \( O\left( (m + \log \log n) \log^2 n \right) \) and succeeds with high probability, which is sufficient to prove Theorem~\ref{thm:main}. 

\subsection{Description}

The algorithm operates in two main phases: (1) recursively decomposing a scaled version of the input graph, and (2) iteratively computing distances on the decomposition.

\subsubsection{Phase 1: Recursive Decomposition}

Let \( G = (V, E, w_G) \) be the input graph. We define a scaled graph \( G' = (V, E, w_{G'}) \) by increasing the weight of every edge by $W/2$, i.e., $w_{G'}(e)=w_G(e)+W/2$ for all edges $e\in E$.

The decomposition process uses the \textsc{Decompose} algorithm from~\cite{bringmann2023negativeweightsinglesourceshortestpaths}, which, given a graph and a diameter parameter \( d \), returns a set of edges whose removal ensures that the remaining strongly connected components (SCCs) are either small in size (condition~\emph{(i)} of Lemma~\ref{lem:decompose}) or have diameter at most \( d/2 \) (which we will deduce from condition \emph{(ii)}). Starting with \( d_0 = nW/2 \), we decompose \( G' \) recursively. For each subgraph \( H \) with parameter \( d \), we compute the set of cut edges \( S_{(H,d)} = \textsc{Decompose}(H,d) \) and remove these edges to obtain SCCs \( C_1, C_2, \dots, C_\ell \). For each SCC \( C_i \), the diameter parameter \( d_i \) is set to \( d \) if \( |C_i| \leq \frac{3}{4}|V(H)| \), or \( d/2 \) otherwise. If \( d_i > W/2 \), we recursively decompose \( C_i \); otherwise, \( C_i \) becomes a leaf node.

Once the decomposition is complete, we check each leaf node for the presence of a negative edge \( uv \). If such an edge exists, then we claim there must be a negative-weight cycle in $G$. To find a negative-weight cycle, we run Dijkstra's algorithm from \( v \) in \( G'_{\geq 0} \) (where negative-weight edges in \( G' \) are replaced by edges of weight 0) to find a path to \( u \) of weight at most $W/2$ in $G'_{\ge 0}$; we will show that such a path must exist. Note that this path also has weight at most $W/2$ in $G$, where edge weights are only smaller. The algorithm then outputs the cycle formed by concatenating this path with the edge \( uv \), which is a negative-weight cycle since the edge $uv$ has negative weight in $G'$ and hence weight less than $-W/2$ in $G$.

At this point, if the algorithm does not terminate early with a negative-weight cycle, then all edges in leaf nodes are non-negative.

\subsubsection{Phase 2: Iterative Distance Computation}

In this phase, we iteratively compute distance estimates on the decomposition tree, starting from the leaves and moving up to the root. Initially, we set \( \phi(v) = 0 \) for all \( v \in V \). For each node \( (H, d) \) in the decomposition tree, we skip further computation if \( (H, d) \) is a leaf, as all edges in \( H \) are non-negative. For a non-leaf node $(H,d)$, assume that we have already computed a $\phi$ such that all SCCs of \( (H \setminus S_{(H, d)})_\phi \) have non-negative weights. Using Lemma~\ref{lem:fix-dag}, we adjust $\phi$ so that all edges of \( (H \setminus S_{(H, d)})_\phi \) now have non-negative weights. 
We then run \textsc{BellmanFordDijkstra} on $H$ with potential function $\phi$, and use $w_{G'_{\ge 0}}$ as an auxiliary function to help detect a negative-weight cycle.

If \textsc{BellmanFordDijkstra} records a path with auxiliary weight more than $d$, then we terminate \textsc{BellmanFordDijkstra} early and recover such a path $P$ with endpoints $u$ and $v$. 
We then run Dijkstra's algorithm from \( v \) in \( G'_{\geq 0} \) to find a shortest path \( P' \) back to \( u \). We show that concatenating \( P \) and \( P' \) produces a negative-weight cycle in $G$. 

After processing all nodes at the current level, we update \( \phi \) and proceed to the next level. After processing all the levels, we guarantee that $G'_\phi$ has non-negative weight edges everywhere.

\subsection{Algorithm Pseudocode}

\begin{algorithm}[H]
\caption{\textsc{Scale} Algorithm}
\begin{algorithmic}[1]
\Function{Scale}{$G = (V, E, w_G), W$}
    \Comment{Input: Graph $G$ with $w_G(e) \geq -W$}
    \ForAll{edges $e \in E$}
        \State $w_{G'}(e) \gets w_G(e) + W/2$
    \EndFor
    \State $G' \gets (V, E, w_{G'})$

    \State \textbf{// Phase 1: Recursive Decomposition}
    \State $d_0 \gets nW/2$
    \State \Call{BuildDecompositionTree}{$G'$, $d_0$}

    \ForAll{leaf nodes $(H, d)$ in the decomposition tree}
        \If{\textbf{exists} edge $uv \in E(H)$ with $w_H(uv) < 0$}
            \State Run Dijkstra's algorithm from $v$ in $G'_{\geq 0}$
            \State Let $P$ be the path from $v$ to $u$ in $G'_{\geq 0}$
            \State \Return negative-weight cycle formed by $P + uv$
        \EndIf
    \EndFor

    \State \textbf{// Phase 2: Iterative Distance Computation}
    \State $\phi \gets \ $\Call{IterativeDistanceComputation}{$G', \text{DecompositionTree}$}

    \State \Return the final potential function $\phi$
\EndFunction
\end{algorithmic}
\end{algorithm}

\begin{algorithm}[H]
\caption{Phase 1: \textsc{BuildDecompositionTree}}
\begin{algorithmic}[1]
\Function{BuildDecompositionTree}{$H$, $d$}
    \State $S_{(H, d)} \gets$ \textsc{Decompose}($H$, $d$)
    \State Remove edges $S_{(H, d)}$ from $H$ to obtain SCCs $C_1, C_2, \dots, C_\ell$
    \ForAll{$i \in [\ell]$}
        \If{$|C_i| \leq \frac{3}{4}|V(H)|$}
            \State $d_i \gets d$
        \Else
            \State $d_i \gets d/2$
        \EndIf
        \If{$d_i > W/2$}
            \State \Call{BuildDecompositionTree}{$C_i$, $d_i$}
        \EndIf
    \EndFor
\EndFunction
\end{algorithmic}
\end{algorithm}

\begin{algorithm}[H]
\caption{Phase 2: \textsc{IterativeDistanceComputation}}
\begin{algorithmic}[1]
\Function{IterativeDistanceComputation}{$G', \text{DecompositionTree}$}
    \State Initialize $\phi(v) \gets 0$ for all $v \in V$
    \ForAll{levels $\ell$ from bottom to top in the decomposition tree}
        \ForAll{nodes $(H, d)$ at level $\ell$}
            \If{$(H, d)$ is a leaf}
                \State \textbf{Skip leaf nodes} \Comment{All edges are non-negative}
            \Else
                \State Adjust $\phi$ so DAG edges are non-negative \Comment{Lemma~\ref{lem:fix-dag}}
                \State Run \textsc{BellmanFordDijkstra}$(H,\phi)$ with auxiliary weight $w_{H_{\ge 0}}$, performing:
                \ForAll{iterations $i$ during \textsc{BellmanFordDijkstra}}
                    \State Compute values $d_i(v)$ and $d'_i(v)$ for all $v\in V$
                    \If{$d'_i(v) > d$ for some vertex $v$}
                        \State Recover the path $P$ from $u$ to $v$ with auxiliary weight exceeding $d$
                        \State Run Dijkstra's algorithm from $v$ in $G'_{\geq 0}$
                        \State Let $P'$ be the path from $v$ to $u$ in $G'_{\geq 0}$
                        \State \Return negative-weight cycle formed by $P + P'$
                    \EndIf
                \EndFor
                \State Update $\phi_{\text{nxt}}(v) \gets$ computed distance for all $v \in V(H)$
            \EndIf
        \EndFor
        \State $\phi \gets \phi_{\text{nxt}}$ \Comment{$H_\phi$ has non-negative weights for each $(H,d)$ on level $\ell$}
    \EndFor
    \State \Return the final potential function $\phi$
\EndFunction
\end{algorithmic}
\end{algorithm}

\subsection{Analysis}

Our ultimate goal is to prove the following theorem.

\begin{restatable}{theorem}{Master}\label{thm:master}
    There is an algorithm \textsc{Scale} such that, for any input graph \( G = (V, E, w_G) \) with edge weights \( w_G(e) \geq -W \) for all \( e \in E \), \textsc{Scale}$(G)$ either returns a reweighted graph \( G_\phi \) with edge weights at least \( -W/2 \), or outputs a negative-weight cycle in \( G \). The algorithm runs in expected time \( O\left( (m + \log\log n) \log^2 n \right) \) and succeeds with high probability.
\end{restatable}

To prove this theorem, we establish several intermediate results. 
Below, we say that a vertex subset $U\subseteq V$ has diameter at most $d$ in a graph $G=(V,E)$ if for any two vertices $u,v\in U$, there is a path in $G$ from $u$ to $v$ of weight at most $d$.
Throughout, we condition on the progress property of Lemma~\ref{lem:decompose} succeeding for all decompositions, which occurs with high probability.

\begin{claim}\label{claim:progress}
    Let $(H,d)$ be a non-leaf node in the decomposition tree, and let \( C \) be a strongly connected component (SCC) of \( H \setminus S_{(H,d)} \) such that \( |C| > \frac{3}{4}|V(H)| \). Then, with high probability, \( C \) has diameter at most \( \frac{d}{2} \) in \( H_{\geq 0} \) (and in \( G'_{\geq 0} \)).
\end{claim}

\begin{proof}
    Consider a SCC $C$ of \( H \setminus S_{(H,d)} \) with \( |C| > \frac{3}{4}|V(H)| \). Lemma~\ref{lem:decompose} guarantees that, for every vertex \( v \in C \), both the in-ball and out-ball of radius \( \frac{d}{4} \) centered at \( v \) in \( H_{\geq 0} \) contain more than \( \frac{|V(H)|}{2} \) vertices. 
    Let \( u, v \in C \). Since both the out-ball of radius \( \frac{d}{4} \) centered at \( u \) and the in-ball of radius \( \frac{d}{4} \) centered at \( v \) each contain more than \( \frac{|V(H)|}{2} \) vertices, their intersection must be non-empty.
    Therefore, there exists $w \in V(H)$ such that there exists a path from \( u \) to \( w \) of length at most \( \frac{d}{4} \) in \( H_{\geq 0} \), and a path from \( w \) to \( v \) of length at most \( \frac{d}{4} \) in \( H_{\geq 0} \).
    By concatenating these two paths, we obtain a path from \( u \) to \( v \) of length at most \( \frac{d}{2} \) in \( H_{\geq 0} \).
    Since \( H_{\geq 0} \subseteq G'_{\geq 0} \), this result extends to \( G'_{\geq 0} \). Therefore, the diameter of \( C \) is at most \( \frac{d}{2} \) in \( H_{\geq 0} \) (and in \( G'_{\geq 0} \)), as desired.
\end{proof}

\begin{claim}\label{claim:diameter}
Let \( (H,d) \) be any node in the decomposition tree with \( d < d_0 \). Then, with high probability, \( V(H) \) has diameter at most \( d \) in \( G'_{\geq 0} \).
\end{claim}

\begin{proof}
    Since \( (H, d) \) is not the root, \( H \) is a strongly connected component (SCC) of its parent \( H' \).
    If \( V(H) > \frac{3}{4}|V(H')| \), then by Claim~\ref{claim:progress}, \( V(H) \) has diameter at most \( d \) in \( H_{\geq 0} \), and hence in \( G'_{\geq 0} \). Otherwise, its parent node \( (H', d) \) has the same diameter parameter \( d < d_0 \). By induction, \( V(H') \) has diameter at most \( d \) in \( G'_{\geq 0} \). Since \( V(H) \) is a subset of \( V(H') \), it also has diameter at most \( d \) in $G'_{\ge 0}$.
    Therefore, in both cases, \( V(H) \) has diameter at most \( d \) in \( G'_{\geq 0} \), as desired.
\end{proof}

\begin{claim}\label{claim:top}
    Let $(H,d)$ be any node in the decomposition tree with \( d = d_0 \). Then, for all simple paths \( P \) in \( H \) with \( w_H(P) \leq 0 \), we have \( w_{H_{\geq 0}}(P) \leq d \).
\end{claim}

\begin{proof}
    Let \( \text{neg}_H(P)\le 0 \) denote the total negative weight (i.e., sum of negative-weight edges) of \( P \) in \( H \). For a simple path \( P \) in \( H \) with \( w_H(P) \leq 0 \), each negative-weight edge in \( H \) has weight at least \( -W/2 \), and \( P \) contains at most \( \lvert V(H) \rvert \leq n \) edges. Therefore, \( \text{neg}_H(P) \geq -W/2 \cdot n = -d_0 \).
    In \( H_{\geq 0} \), all negative-weight edges are replaced with edges of weight 0. Hence, the weight of \( P \) in \( H_{\geq 0} \) is given by \( w_{H_{\geq 0}}(P) = w_H(P) + \lvert \text{neg}_H(P) \rvert \). Substituting \( w_H(P) \leq 0 \) and \( \lvert \text{neg}_H(P) \rvert \leq W/2 \cdot n = d \), we find \( w_{H_{\geq 0}}(P) \leq d \), as desired.
\end{proof}

\begin{claim}\label{claim:nonegcycles}
    Let $(H,d)$ be a node in the decomposition tree with \( d < d_0 \). With high probability, if there exists a path \( P \) in \( H \) such that \( w_H(P) \leq 0 \) but \( w_{H_{\geq 0}}(P) > d \), then  concatenating $P$ with a shortest path $P'$ in $G'_{\ge 0}$ from the end of $P$ back to its start produces a negative-weight cycle in $G$. In particular, if \( G \) has no negative-weight cycles, then for all (potentially non-simple) paths \( P \) in \( H \) with \( w_H(P) \leq 0 \), we have \( w_{H_{\geq 0}}(P) \leq d \).
\end{claim}

\begin{proof}
    Suppose there exists a path \( P \) in \( H \) such that \( w_H(P) \leq 0 \) but \( w_{H_{\geq 0}}(P) > d \). Let \( \text{neg}_H(P) \) denote the total negative weight of \( P \) in \( H \). The weight of \( P \) in \( H_{\geq 0} \) is given by \( w_{H_{\geq 0}}(P) = w_H(P) + \lvert \text{neg}_H(P) \rvert \). Since \( w_H(P) \leq 0 \) and \( w_{H_{\geq 0}}(P) > d \), it follows that \( \lvert \text{neg}_H(P) \rvert > d \).
    For each negative-weight edge $e$ in \( H \), its weight in $G$ is $w_G(e)=w_H(e)-W/2\le w_H(e)-|w_H(e)|$, so the weight of \( P \) in \( G \) is \( w_G(P) \le w_H(P) - \lvert \text{neg}_H(P) \rvert \). Substituting \( w_H(P) \leq 0 \) and \( \lvert \text{neg}_H(P) \rvert > d \), we find \( w_G(P) < -d \).
    By Claim~\ref{claim:diameter}, \( V(H) \) has diameter at most \( d \) in \( G'_{\geq 0} \), so the shortest path \( P' \) in \( G'_{\geq 0} \) from the end of \( P \) back to its start satisfies \( w_{G'_{\geq 0}}(P') \leq d \). Since edge weights can only be smaller in $G$, we also have $w_G(P')\le d$. Combining \( P \) and \( P' \) forms a cycle \( C = P + P' \) in \( G \) with \( w_G(C) = w_G(P) + w_G(P') < -d + d = 0 \), so $C$ is a negative-weight cycle.
\end{proof}

\begin{claim}\label{claim:sparsehitting}
    Let $(H,d)$ be a non-leaf node in the decomposition tree. Then, for any (potentially non-simple) path \( P \) in \( H \) with \( w_{H_{\geq 0}}(P) \leq d \), the number of edges of $P$ inside $S_{(H,d)}$ (counting multiplicity) has expectation at most \( O(\log n) \).
\end{claim}

\begin{proof}
Follows directly from Lemma~\ref{lem:decompose}.
\end{proof}

We now proceed to analyze the runtime and correctness of the \textsc{Scale} algorithm by considering whether $H$ contains a negative cycle. Below, let \( n_H = |V(H)| \) and \( m_H = |E(H)| \) denote the number of vertices and edges in $H$.

\begin{lemma}\label{lemma:neg_cycle_in_G}
    Consider a node $(H,d)$, and suppose \( H \) does not contain a negative-weight cycle. Then, \textsc{BellmanFordDijkstra}$(H,\phi)$ runs in expected time \(O\left( (m_H + n_H \log \log n_H) \log n \right) \).
\end{lemma}

\begin{proof}
    There are two cases to consider:

    \textbf{Case 1:} For each vertex $v\in V(H)$, there exists a shortest path to $v$ of weight at most $d$ in $H_{\ge 0}$. In this case, the shortest path $P$ to $v$ with \( w_{H_{\geq 0}}(P) \leq d \) has an expected $O(\log n)$ edges inside $S_{(H,d)}$ by Claim~\ref{claim:sparsehitting}. Since $H_\phi$ has non-negative weight edges outside of $S_{(H,d)}$, we have $\mathbb E[\eta_{H,\phi}(v)]\le O(\log n)$ for the parameter $\eta_{H,\phi}$ defined in Lemma~\ref{lem:BFD}. By Lemma~\ref{lem:BFD}, \textsc{BellmanFordDijkstra} runs in expected time \( O(\sum_v(\deg_H(v)+\log\log n)\cdot\eta_{H,\phi}(v)), \) which has expectation \( O((m_H + n_H \log \log n_H) \log n) \) since \( \mathbb E[\eta_{H,\phi}(v)]\le O(\log n) \) for all $v\in V(H)$. Note that \textsc{BellmanFordDijkstra} can still terminate early and output a negative-weight cycle in $G$, but that only speeds up the running time.

    \textbf{Case 2:} For some vertex $v\in V(H)$, every shortest path to $v$ has weight more than $d$ in $H_{\ge 0}$. By Claim~\ref{claim:nonegcycles}, this implies the existence of a negative-weight cycle in \( G \). Among all shortest paths to $v$, let $Q$ be one minimizing $w_{H_{\ge 0}}(Q)$, which is still greater than $d$. Take the longest prefix $Q'$ of $Q$ with weight at most $d$ in $H_{\ge 0}$. Since $w_{H_{\ge 0}}(Q')\le d$, the number of edges of $Q'$ inside $S_{(H,d)}$ has expectation \( O(\log n) \) by Claim~\ref{claim:sparsehitting}. Let $Q''$ be the path $Q'$ concatenated with the next edge in $Q$, so that $w_{H_{\ge 0}}(Q'')>d$ and the number of edges of $Q''$ inside $S_{(H,d)}$ still has expectation \( O(\log n) \). Since $Q''$ is a prefix of shortest path $Q$, it is also a shortest path. Let $u\in V(H)$ be the endpoint of $Q''$, and let $i$ be the number of iterations of \textsc{BellmanFordDijkstra} before $d_i(u)= w_H(Q'')$, so that $i$ has expected value $O(\log n)$. After $i$ iterations, \textsc{BellmanFordDijkstra} finds a shortest path $P$ ending at $u$ with $w_H(P)=d_i(u)=w_H(Q'')$. Moreover, $w_{H_{\ge 0}}(P)\ge w_{H_{\ge 0}}(Q'')$ since otherwise, replacing the prefix $Q''$ of $Q$ by $P$ results in a shortest path to $v$ of lower weight in $H_{\ge 0}$, contradicting the assumption on $Q$. It follows that after an expected $O(\log n)$ iterations of \textsc{BellmanFordDijkstra}, a path $P$ is found with $w_H(P)\le 0$ and $w_{H_{\ge 0}}(P)>d$, and the algorithm terminates with a negative-weight cycle. The runtime is $O((m_H + n_H \log \log n_H) \cdot i)$, which has expectation $O((m_H + n_H \log \log n_H)\log n)$.
\end{proof}

\begin{lemma}\label{lemma:neg_cycle_in_G_prime}
Consider a node $(H,d)$, and suppose \( H \) contains a negative-weight cycle. Then, \textsc{BellmanFordDijkstra}$(H,\phi)$ runs in expected time \(O\left( (m_H + n_H \log \log n_H) \log n \right) \).
\end{lemma}

\begin{proof}
    There are two cases to consider:

    \textbf{Case 1:} \( H \) contains a negative-weight cycle consisting entirely of non-positive-weight edges.  
    Since \textsc{Decompose} only removes positive-weight edges, such a cycle remains intact throughout the decomposition and persists in one of the leaf nodes. A negative-weight edge from this cycle is found in this leaf node, and the algorithm computes a negative-weight cycle and terminates early. In particular, the node $(H,d)$ is never processed. 

    \textbf{Case 2:} All negative-weight cycles in \( H \) include at least one positive-weight edge. This case is the most subtle and requires carefully defining a reference path $Q$.

Define $M=-(d+1)\cdot|V(H)|\cdot W$. We first claim that any (potentially non-simple) path in $H$ of weight at most $M$ in $H$ must have weight exceeding $d$ in $H_{\ge 0}$. Given such a path, we first remove any cycles on the path of non-negative weight in $H$. The remaining path still has weight at most $M=-(d+1)\cdot|V(H)|\cdot W$, so it must have at least $(d+1)\cdot|V(H)|$ edges, which means it can be decomposed into at least $d+1$ negative-weight cycles, each with at least one positive-weight edge in $H_{\ge 0}$. It follows that the path has weight at least $d+1$ in $H_{\ge 0}$, concluding the claim.

Among all paths of weight at most $M$ in $H$, let $Q$ be one with minimum possible weight in $H_{\ge 0}$, which must be greater than $d$. Take the longest prefix $Q'$ of $Q$ with weight at most $d$ in $H_{\ge 0}$. Since $w_{H_{\ge 0}}(Q')\le d$, the number of edges of $Q'$ inside $S_{(H,d)}$ has expectation \( O(\log n) \) by Claim~\ref{claim:sparsehitting}. Let $Q''$ be the path $Q'$ concatenated with the next edge in $Q$, so that $w_{H_{\ge 0}}(Q'')>d$ and the number of edges of $Q''$ inside $S_{(H,d)}$ still has expectation \( O(\log n) \). Let $u\in V(H)$ be the endpoint of $Q''$, and let $i$ be the number of iterations of \textsc{BellmanFordDijkstra} before $d_i(u)\le w_H(Q'')$, so that $i$ has expected value $O(\log n)$. After $i$ iterations, \textsc{BellmanFordDijkstra} finds a path $P$ ending at $u$ with $w_H(P)=d_i(u)\le w_H(Q'')$. Moreover, $w_{H_{\ge 0}}(P)\ge w_{H_{\ge 0}}(Q'')$ since otherwise, replacing the prefix $Q''$ of $Q$ by $P$ results in a path of lower weight in $H_{\ge 0}$ and of weight at most $M$ in $H$, contradicting the assumption on $Q$. It follows that after an expected $O(\log n)$ iterations of \textsc{BellmanFordDijkstra}, a path $P$ is found with $w_H(P)\le 0$ and $w_{H_{\ge 0}}(P)>d$, and the algorithm terminates with a negative-weight cycle. The runtime is $O((m_H + n_H \log \log n_H) \cdot i)$, which has expectation $O((m_H + n_H \log \log n_H)\log n)$.
\end{proof}

We can now restate and prove Theorem~\ref{thm:master}.
\Master*

\begin{proof}
    We first consider the decomposition phase. For each node \( (H, d) \), the decomposition takes expected time \( O((m_H + n_H \log \log n_H) \log n) \) by Lemma~\ref{lem:decompose}. At each level of the decomposition tree, the total number of edges and vertices across all nodes is bounded by \( m \) and \( n \), respectively, so the expected runtime per level is \( O((m + n \log \log n) \log n) \). With \( O(\log n) \) levels in the decomposition tree, the total expected runtime is \( O((m + n \log \log n) \log^2 n) \).

    For processing each node $(H,d)$, Lemma~\ref{lem:fix-dag} adjusts $\phi$ in $O(n_H+m_H)$ time, and Lemmas~\ref{lemma:neg_cycle_in_G} and \ref{lemma:neg_cycle_in_G_prime} establish an expected runtime of \( O((m_H + n_H \log \log n_H) \log n) \) for \textsc{BellmanFordDijkstra}. There may be an additional $O(m+n\log\log n)$ time to output a negative-weight cycle, but this only happens once. Again, summing over all nodes $(H,d)$, the total expected runtime is \( O((m + n \log \log n) \log^2 n) \).

    To convert distances in \( G' \) into a reweighted graph \( G_\phi \) with edge weights at least \( -W/2 \), we note that setting \( \phi \) as the distances in $G'$ makes \( G'_\phi \) non-negative. Since \( w_{G'}(e) = w_G(e) + W/2 \) for all $e\in E$, it follows that \( w_{G_\phi}(e) = w_{G'_\phi}(e)-w_{G'}(e)+w_G(e) \geq -W/2 \), as desired.
\end{proof}

\bibliographystyle{alpha}
\bibliography{references}

\end{document}